\documentclass[letterpaper,10pt.]{article}
\usepackage{amsfonts}
\usepackage{amsmath}
\usepackage{amsthm}
\usepackage{amssymb}
\usepackage{algorithm}
\usepackage{framed}
\usepackage[noend]{algpseudocode}
\usepackage[hyphens]{url}
\usepackage{hyperref}
\usepackage{graphicx}
\usepackage{caption}
\usepackage{subcaption}
\usepackage{enumitem}
\usepackage{xspace}
\usepackage{xcolor}
\usepackage{multirow}
\usepackage{rotating}
\usepackage{float}

\usepackage[
  margin=2cm,
  includefoot,
  footskip=30pt,
]{geometry}
\usepackage{layout}

\makeatletter\renewcommand{\ALG@name}{Protocol}\makeatother

\definecolor{orange}{rgb}{1,0.5,0}

\newtheoremstyle{mystyle}{}{}{\itshape}{}{\bfseries}{.}{6 pt}{\thmname{#1}\thmnumber{ #2}\thmnote{ {\bfseries(#3)}}}

\theoremstyle{mystyle} 

\newtheorem{theorem}{Theorem}

\newtheorem{lemma}{Lemma}
\newtheorem{observation}{Observation}

\newtheorem{definition}{Definition}

\begin{document}

\title{Population Protocols with  Faulty  Interactions: the Impact of a Leader\footnotemark[1]}
\author{Giuseppe Antonio Di Luna\footnotemark[2]  \hspace{5mm} Paola Flocchini\footnotemark[2] \hspace{5mm} 
Taisuke Izumi\footnotemark[3] \\ Tomoko Izumi\footnotemark[4] \hspace{5mm} Nicola Santoro\footnotemark[5]
\hspace{5mm} Giovanni Viglietta\footnotemark[2]}
\date{}

\clearpage\maketitle
\thispagestyle{empty}
\renewcommand{\thefootnote}{\fnsymbol{footnote}}
\footnotetext[1]{This work was supported in part by  NSERC Discovery Grants,  and by
KAKENHI No.\ 15H00852 and 25289227.}
\footnotetext[2]{University of Ottawa. E-mails: \{gdiluna, paola.flocchini, gvigliet\}@uottawa.ca.}
\footnotetext[3]{Nagoya Institute of Technology, Gokiso-cho, 
Showa-ku, Nagoya, Aichi, 466-8555, Japan. E-mail: t-izumi@nitech.ac.jp.}
\footnotetext[4]{Ritsumeikan University. E-mail: izumi-t@fc.ritsumei.ac.jp.}
\footnotetext[5]{Carleton University. E-mail: santoro@scs.carleton.ca.}

\renewcommand{\thefootnote}{\arabic{footnote}}
\setcounter{footnote}{0}


\begin{abstract}
\noindent We consider the problem of simulating traditional population protocols under weaker models of communication, which include one-way interactions (as opposed to two-way interactions) and omission faults (i.e., failure by an agent to read its partner's state during an interaction), which in turn may be detectable or undetectable. We focus on the impact of a leader, and we give a complete characterization of the models in which the presence of a unique leader in the system allows the construction of simulators: when simulations are possible, we give explicit protocols; when they are not, we give proofs of impossibility. Specifically, if each agent has only a finite amount of memory, the simulation is possible only if there are no omission faults. If agents have an unbounded amount of memory, the simulation is  possible as long as omissions are detectable. If an upper bound on the number of omissions involving the leader is known, the simulation is always possible, except in the one-way model in which one side is unable to detect the interaction.
\end{abstract}

\section{Introduction} 

\subsection{Framework}

Consider a system of {\em interacting computational  entities}, called agents,  whose interaction is  however not under their control 
but decided by an external scheduler.
Such are for example systems of   wireless  mobile entities  where
two entities  can interact (i.e., exchange information)  when their movement 
brings them  into communication range of each other, but 
  their  movements, and thus  their interactions,  are  unpredictable.
Systems satisfying this condition, sometimes  called {\em opportunistic mobility} or {\em passive mobility}, have been 
extensively examined under  variety of assumptions, especially  within the  context of distributed computing
in  highly dynamic networks and time-varying graphs (for recent surveys see \cite{CaFQS12,popbook}).

In particular, in the  {\em population protocol} model ({\sf PP}),  introduced in the seminal paper \cite{first}, the entities
are assumed to be finite-state and anonymous (i.e., identical), execute the same protocol, and interactions are always between pairs of agents. 
The roles of the two agents involved  in an interaction  are asymmetric: one agent is considered the {\em starter} and the other is the {\em reactor}. Still,
 the  communication is  
 {\em two-way}: 
each agent  receives  the state of the other  and  
executes the protocol to update its own state  based on the received information and its own state. Furthermore, in the selection
of the occurrences of the interactions,
 the scheduler is constrained to satisfy some  fairness assumption.
 
The restricted computational universe defined by the basic assumptions of  {\sf PP} has been subsequently expanded in
an attempt  to overcome the inherent computability  limitations and to examine the computational impact of  factors such as
non-constant memory 
(e.g.,  \cite{dalistar2,dalistar,passivemachine}),   presence of  a leader (e.g., \cite{fastleader}), 
storage of information  on edges  (e.g., \cite{mediated,mediated3,mediated2}), etc.

In all these models, including the original one, the interaction is assumed to be fault-free.
 An immediate important question is  what happens  if  interactions are subject to {\em  failures}.

Very little is know in this regard. An insight comes from the study of the so-called {\em one-way interaction models}~\cite{oneway}, 
where the starter of an interaction  is not  able to see 
the state of  the reactor ({\em immediate transmission}), or it is not even able to detect that the
interaction has taken place  ({\em immediate observation}). 
This study showed that, under one-way interactions, the computational power of  the agents  is strictly
weaker  than with the usual bidirectional interactions.  In particular, if  the interactions are not detectable by the starter (i.e., immediate observation),
the agents can compute  only the threshold predicates \cite{oneway}.  

The one-way interaction models capture a class 
of  {\em permanent omission failures}, those occurring at the starter's side.
Clearly, there are many more types of omission failures, such as those occurring at the reactor's side and, 
more insidious, those whose occurrence is dynamic and unpredictable. And of course,
for each of these types there are different variations, depending on the kind of fault detection assumed.

The complete range of {\em dynamic omission failures} has been classified in \cite{DiFIISV16}, where the following general question 
was posed: under what additional system capabilities is it possible to correctly execute every traditional two-way population protocol
in spite of dynamic omission failures?  More specifically, under what conditions (if any)
it is possible to  {\em simulate} the execution of every two-way population protocol for a given class of omission failures?
The simulator   should be a  population protocol  that,  in each execution   in the model defined by the considered  class of omissions,
 produces a correct execution of   any  traditional
 two-way population protocol  $\mathcal P$   given as input,
  regardless of the nature of $\mathcal P$ and the constraints its execution might have;
  and it does so unobtrusively at each agent,  interacting with $\mathcal P$ only by observing its internal state, 
  and  providing to it the internal state of another agent.
 In other words, a simulator provides an  interface between the simulated protocol $\mathcal P$ 
 and the physical communication layer, 
  giving the system the illusion of being in a fault-free  two-way environment.

The existence of simulators is important in scenarios in which we do not only concern ourselves with the final output of a population protocol, but also with the \emph{execution} that leads to the result. We may want, for instance, to guarantee that our simulating agents enter some critical states exactly as many times as they would if they were actually executing the protocol that is being simulated.

The existence of fault-tolerant one-way simulators of two-way protocols has been investigated in  \cite{DiFIISV16} in terms of the amount {\em memory} required by the agents
to perform such simulations, and a variety of models and results were established.  It is shown that, with no \emph{a-priori} knowledge, the
simulation of two-way protocols in the presence of omissions is impossible
even if the agents have infinite memory. In the weakest models investigated, this impossibility holds even if the number of omission failures in each execution is limited to one.
On the other hand, it is also shown that simulation is possible if agents have unique IDs or the total number of agents is known.
Moreover, in some restricted models, simulation is possible when an upper bound on the number of omission faults is known.

In this paper we continue this general line of research  and investigate how the presence in the system 
of a distinguished agent, a {\em leader}, can impact the capability of the system to
tolerate dynamic omission failures. More precisely, we study the possibility and impossibility of
  simulation of two-way protocols with the aid of a leader,  with respect to the different classes of omission failures and one-way interactions.

\subsection{Main Contributions}
As in \cite{DiFIISV16}, we consider 
 all the computationally distinct
models that arise  
from the introduction of omission faults and/or 
one-way interactions in two-way  protocols: {\sf TW}, {\sf IT}, {\sf IO},  {\sf T$_i$} ($i=1,2,3,$),  and {\sf I$_j$} ($j=1,2,3,4,$);
 see Figure~\ref{figure:modcompress}, where the transition function $\delta$, detailed in Section~\ref{sec:model}, uniquely identifies each model.
In particular, {\sf TW} refers to two-way protocols without omissions; {\sf IT} and {\sf IO} refer to the one-way models {\em immediate transmission} and {\em immediate observation}, introduced in \cite{oneway};  the {\sf T$_i$}'s  and {\sf I$_i$}'s refer to the distinct two-way and one-way models with omissions,  respectively. 

We consider  two types of  omission adversaries: informally, a ``malignant'' one ({\bf UO}), which is able to arbitrarily insert omission faults into ``globally fair'' sequences of interactions, and a ``benign" one ({\bf $\diamond$NO}), which inserts some omission faults, but eventually stops. To make our results stronger, we always assume the benign adversary in in the impossibility proofs and the malignant one in the possibility proofs.
  
We study the negative impact that omissions have on computability, and we show that the  simulation of two-way protocols is impossible even with the aid of a leader  (Theorem \ref{th:ledmemfinimp}), assuming that the amount of memory is bounded.

On the other hand, we show that  the presence of  both a leader and infinite memory on each agent makes the simulation possible  in the weak  intermediate one-way models {\sf I$_1$} and {\sf I$_2$} (Theorem~\ref{th2:leaderinfmemsim}), and thus in all the upper  models of Figure~\ref{figure:modcompress}. The fact that this possibility does not apply to {\sf IO} and {\sf T$_1$} is not accidental: indeed we prove that, for these two models,
  the simulation is impossible even with both a leader and infinite memory, even against the  benign omission adversary (Theorem~\ref{th:infmemleader}).

Finally, we study what happens when a bound on the omission failures involving the leader is known, and essentially we show that simulators exists for models  {\sf I$_1$} and {\sf I$_2$} (Theorem~\ref{upbound:1}) and model {\sf T$_1$} (Theorem~\ref{upbound:2}), and these imply the possibility of simulations in all other omissive models. 

For non-omissive models,
we  show that two-way simulation is possible
 in the {\sf IT} model (Theorem~\ref{leader:token}).  In light of the fact  that 
with constant memory, in absence of additional capabilities,  {\sf IT} protocols are strictly less powerful than {\sf TW} (see~\cite{DiFIISV16}), our results show that this computational gap can
be overcome by using a leader.

Our main results are summarized in Figure~\ref{id:algorwq2}, where white blobs represent possibilities, and gray blobs impossibilities. As a consequence of these results, we have a complete characterization
 of the feasibility of simulations  with respect to the presence of a leader.


\subsection{Related Work}

Starting with the seminal paper~\cite{first2}, there have been  extensive  investigations  on population protocols (e.g., see~\cite{computability,infinitepopulation,scheduler,popbook,probabilistic,popchecm,ids}). In order to overcome the inherent computability restrictions of the model, several extensions have been proposed. 
For example, endowing each agent with non-constant memory~\cite{dalistar2,dalistar,passivemachine}, assuming the presence of  a leader \cite{fastleader}, allowing 
 a certain amount of information to be stored on the edges of the ``communication graph''~\cite{mediated,mediated3,mediated2}, etc.

The possibility of reliable computations in {\sf PP}, first considered in   \cite{firstfault}, has been studied only with respect to processors' faults, and the basic model has necessarily been expanded.
In \cite{fault} it has been shown how to compute functions tolerating ${\cal O}(1)$ crash stops and transient failures,
assuming that  the number of failures is bounded and known.
In \cite{majbiz}  the majority problem under ${\cal O}(\sqrt{n})$ Byzantine failures,  
 assuming a fair probabilistic scheduler, has been studied. 
In \cite{bizfault} unique IDs are assumed, and it is shown how to compute functions tolerating a bounded number of Byzantine faults, under the assumption that Byzantine agents cannot forge IDs. 
Self-stabilizing solutions have been devised for specific problems such as: 
 leader election,  assuming knowledge of the system size  and a non-constant number of states~\cite{izumile},
                       or assuming a leader detection oracle~\cite{leaderelection};
 counting,   assuming the presence of  a   leader   \cite{spaceoptcounting}. 
 Moreover,  in \cite{Beauquier20114247}  a self-stabilizing transformer for general protocols has been studied in a slightly different model and under the assumption of unbounded memory and a leader. 

Finally, to the best of our knowledge,  the one-way model, without omissions, has been studied   only in \cite{oneway}, where it is shown that {\sf IT} and {\sf IO}, when equipped with constant memory, can compute a set of functions that is strictly included in that of {\sf TW}.  
The omission models that we consider have been introduced for the first time in \cite{DiFIISV16}, where a characterisation of what can be simulated without a leader is given. Our paper complements and enriches the results of \cite{DiFIISV16}, showing what additional power is obtained assuming the presence of a leader.  


%

\section{Model and Terminology}
\label{sec:model}

In this section we briefly define the computation model, the notion of omission, and the notion of simulator. Due to space constraints, we do not include all the formal definitions, which can be found in~\cite{DiFIISV16}.

\subsection{Interacting Entities}
We consider a system consisting of a set $A = \{a_{1},\ldots,a_{n}\}$ of interacting computational entities, called {\em agents}. 
Each {\em  interaction} involves only  two agents
with asymmetric roles: one agent is the  \emph{starter} and the other is the \emph{reactor}. Interactions occur at discrete times, and at every ``time unit'' exactly one interaction occurs. The starter and the reactor of each interaction are chosen by an external ``adversarial scheduler'' in a ``globally fair'' way (see~\cite{DiFIISV16} for details).

 When two agents interact, they exchange information and perform a local computation according to the same protocol $\mathcal{P}$.
 A protocol is a pair $\mathcal{P} = (Q_{\mathcal{P}},\delta_{\cal P})$, where
 $Q_{\mathcal{P}}$ is  a set of local states and  $\delta_{\cal P}\colon Q_{{\cal P}} \times Q_{{\cal P}} \rightarrow 
Q_{{\cal P}} \times Q_{{\cal P}}$ is the transition function 
defining the states of the two interacting agents at the end of their local 
computation. Some elements of $Q_{\mathcal{P}}$ are labeled as ``initial states''; when the execution of the protocol begins, all agents have (any combination of) initial states.  With a small abuse of notation, and when no ambiguity arises, we will use the same literal (e.g., 
$a_i$) to indicate both an agent and its local
state. A \emph{configuration} of $\mathcal P$ is a multiset of local states of $\mathcal P$.

We can model the presence of a \emph{leader} in the system by stipulating that, in every initial configuration, there is exactly one agent in a distinguished state (or set of states).

Depending on the conditions imposed on the transition function, 
three main models of interactions have been  identified: the standard {\em two-way} model
and  the one-way models,
\emph{immediate transmission}  and  \emph{immediate observation}, presented in \cite{oneway}.

 {\bf Two-Way Interaction Model} (\textsf{TW}).  The state transition function consists of
two functions $f_s \colon Q_{\mathcal{P}} \times Q_{\mathcal{P}} \to Q_{\mathcal{P}}$ and 
$f_r\colon Q_{\mathcal{P}} \times Q_{\mathcal{P}} \to Q_{\mathcal{P}}$, one for the starter and the other for the receiver respectively, with
$\delta_{\mathcal{P}}(a_s,a_r) = (f_s(a_s,a_r), f_r(a_s,a_r))$.
\smallskip

 {\bf Immediate Transmission Model} (\textsf{IT}).  The state transition function consists of 
two functions $g \colon Q_{\mathcal{P}} \to Q_{\mathcal{P}}$ and 
$f\colon Q_{\mathcal{P}} \times Q_{\mathcal{P}} \to Q_{\mathcal{P}}$, with
$\delta_{\mathcal{P}}(a_s,a_r) = (g(a_s), f(a_s,a_r))$.
Note that, in the {\sf IT} interaction model, the starter does not read the state of the reactor but it explicitly detects the interaction, as it applies function $g$ to its own state. 
\smallskip

 {\bf Immediate Observation Model} (\textsf{IO}). The state transition function has the form 
$\delta_{\mathcal{P}}(a_s, a_r) = (a_s, f(a_s,a_r))$.
 Note that, in the {\sf IO} model,  there is no  detection of the interaction (or proximity) by the starter.

\subsection{Omissions}

An {\em omission} is a fault involving a single interaction. In an omissive interaction, an agent does not receive any information about the state of the other.
If omissions can occur in the system,  then transition functions become more general relations. 

\smallskip
 {\bf Two-Way Omissive Models.} In the most general omissive model, {\sf T$_3$}, the transition relation has the form
$$\delta(a_s,a_r)=\{(f_s(a_s,a_r),f_r(a_s,a_r)),\ (o(a_s),f_r(a_s,a_r)),\ (f_s(a_s,a_r),h(a_r)),\ (o(a_s),h(a_r))\}.$$ 
The first pair is the outcome of an interaction when no omission is present; the other three pairs represent all possible outcomes when there is an omission: respectively, an omission on the starter's side, on the reactor's side, and on both sides. The functions $o$ and $h$ represent the detection capabilities of each agent: if one of these is the identity, then omissions are \emph{undetectable} on the respective side. This gives rise to the weaker models {\sf T$_2$} and {\sf T$_1$} depicted in Figure~\ref{figure:modcompress} (see~\cite{DiFIISV16} for more details).
\smallskip

 {\bf One-Way Omissive Models.}
These models are defined by the 
 transition relation
$$\delta(a_s,a_r)=\{(g(a_s),f(a_s,a_r)),\ (o(a_s),h(a_r))\}.$$
 The first pair is the outcome of an interaction when no omission is present, and the second pair when there is an omission. 
Note that the {\sf IO} model corresponds to the case in which $g$ is the identity function and there are no omissions. Once again, omissions are undetectable starter-side if $o$ is the identity function or if $o=g$. Moreover, if $h=g$, the reactor has detected the \emph{proximity} of another agent, but is unable to read its state or even determine who is the starter and who is the reactor. Collectively, these variations give rise to models {\sf I$_1$} to {\sf I$_4$} in Figure~\ref{figure:modcompress}. Other combinations of omissions and detections are possible, but they are provably equivalent to some of the aforementioned ones (see~\cite{DiFIISV16} for more details).

 Omissions are introduced by an adversarial entity. We consider two types of adversaries:\\
(1) the {\em Unfair Omissive Adversary} ({\bf UO}), which arbitrarily inserts omissive interactions in any execution, and\\
(2) the {\em Eventually Non-Omissive Adversary} ($\diamond${\bf NO}), which can only insert finitely many omissions in an execution.

\begin{figure}
\begin{center}
\includegraphics[scale=0.65]{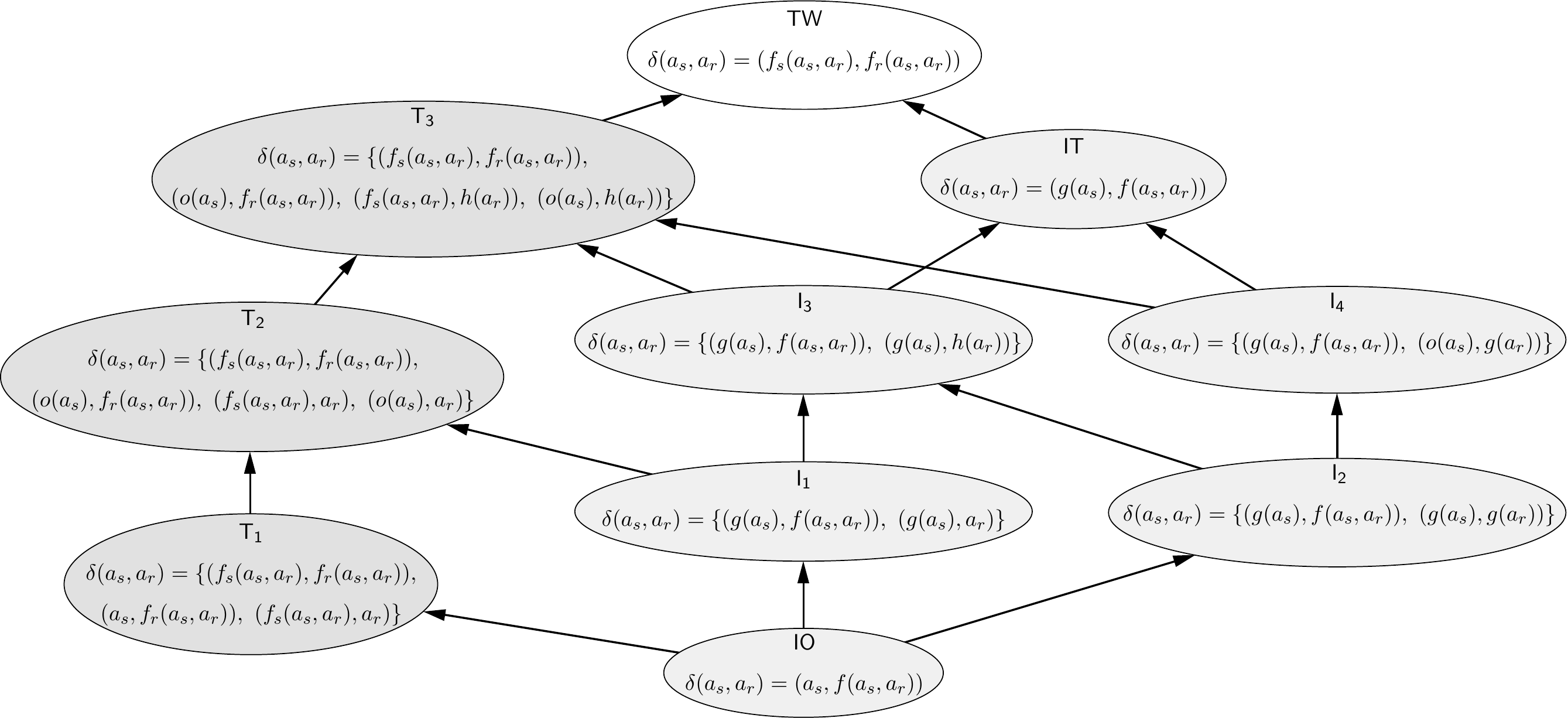}
\end{center}
\caption{Interaction models (up to equivalence) and their computational relationships. An arrow between two blobs indicates that the class of solvable problems in the source blob is included in that of the destination blob. The models on the left,  {\sf T$_1$}, {\sf T$_2$}, {\sf T$_3$}, are the two-way models with omissions. The models on the right, {\sf I$_1$}, {\sf I$_2$}, {\sf I$_3$}, {\sf I$_4$},
are the one-way models with omissions. \label{figure:modcompress}}
\end{figure}

\subsection{Simulation of Two-Way Protocols}
\label{s:simulation}
Let $\mathcal P$ be a two-way protocol, and let $\mathcal{S}(\mathcal{P})$ be any protocol (which could be one-way, omissive, or both). Next we are going to informally define what it means for $\mathcal{S}(\mathcal{P})$ to \emph{simulate} $\mathcal P$ (for a formal definition, refer to~\cite{DiFIISV16}).

We want the set of local states of $\mathcal{S}(\mathcal{P})$ to be of the form $Q_{\mathcal{P}} \times Q_{\mathcal{S}}$, where $ Q_{\mathcal{P}}$ is the set of local states of $\mathcal P$ (the ``simulated states''), and $ Q_{\mathcal{S}}$ is some additional memory space used in the simulation. Suppose now to start an execution of $\mathcal{S}(\mathcal{P})$ on a system of $n>2$ agents from a given initial configuration. Agents are allowed to freely change the $Q_{\mathcal{S}}$ component of their local states; but when they change their $Q_{\mathcal{P}}$ component, we want the change to reflect the transition function of $\mathcal{P}$. That is, if $\delta_{\mathcal{P}}(a_s,a_r) = (f_s(a_s,a_r), f_r(a_s,a_r))$, then for every agent whose simulated state changes from $a_s$ to $f_s(a_s,a_r)$, there must be some other agent (at some point in time) whose simulated state changes from $a_r$ to $f_r(a_s,a_r)$. Moreover, there must be a perfect matching between such transitions, in such a way that each starter of a simulated transition can be implicitly mapped to an appropriate reactor. Also, such a perfect matching must be ``temporally consistent'', i.e., there must be an ordering of the simulated two-way interactions that respects the order of the local state changes of each agent.

We additionally require that, if the execution of $\mathcal{S}(\mathcal{P})$ is globally fair (in the sense defined in~\cite{DiFIISV16}), then also the resulting simulated execution of $\mathcal{P}$ is globally fair.


\begin{figure}[H]
\begin{center}
\includegraphics[scale=0.75]{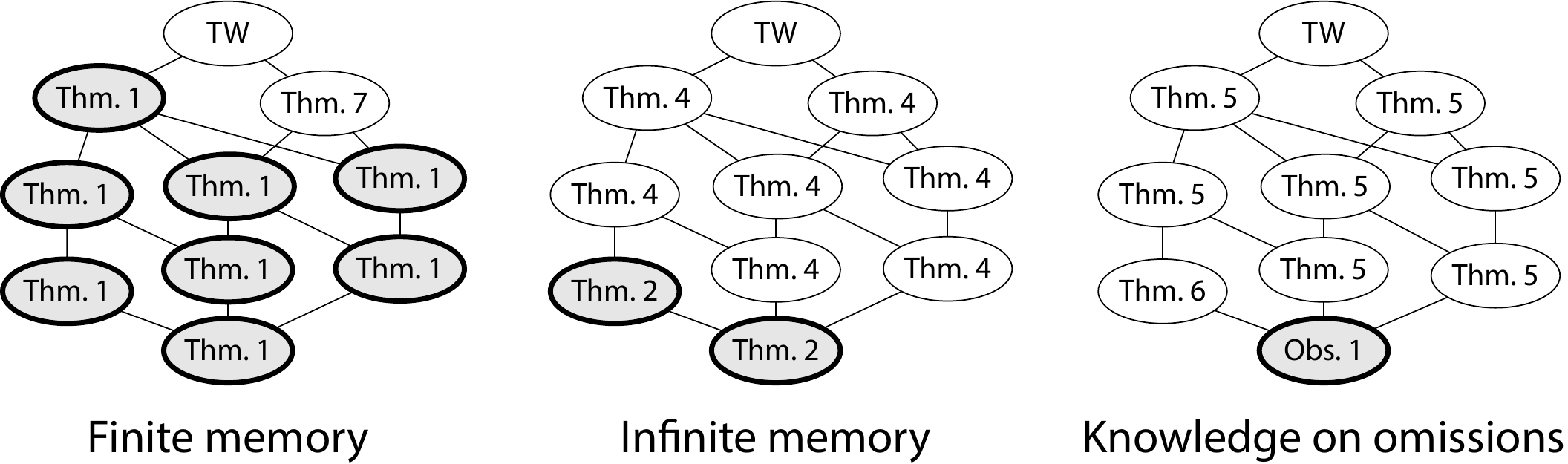}
\end{center}
\caption{Map of results (cf.~Figure~\ref{figure:modcompress})\label{id:algorwq2}. White blobs denote the existence of simulators; gray blobs indicate that simulations are not possible.}
\end{figure}


\section{Simulation with a Leader  in Omissive Models: Impossibility \label{imp1}}

In this section we prove that
the presence of a leader, alone, might not be  sufficient to overcome dynamic omissions. Indeed, we  prove that
there are two-way protocols that cannot be simulated with omissive interactions even if a leader is present.
%
%

Next we consider the {\em Pairing Problem} introduced in~\cite{DiFIISV16}:
a set of agents $A$ is given,  partitioned into  \emph{consumer} agents $A_c $,  starting in state   $c$, 
and \emph{producer} agents  $A_p$, starting in state $p$. 
We say that a protocol $\mathcal P$ solves the \emph{Pairing Problem} if it enforces the following properties:
\begin{enumerate}
\item[(i)] {\em Irrevocability:} $\mathcal P$ has a state $cs$ that only agents in state $c$ can reach; once an agent has state $cs$, its state cannot change any more;
\item[(ii)] {\em Safety:} At any time, the number of agents in  state $cs$ is at most $|A_p|$;
\item[(iii)] {\em Liveness:} In all globally fair executions of $\mathcal P$, eventually the number of agents in  state $cs$ is stably equal to $\min\{|A_c|,|A_p|\}$.
\end{enumerate}

This problem can be solved in the standard two-way model by the  simple protocol below:

\begin{framed}

\footnotesize
{\sf  Pairing  Protocol} $\mathcal{P}_{IP}$.
$Q_{\mathcal{P}_{IP}} = \{cs,c,p,\bot\}$.  The only non-trivial transition rules are $(c,p) \mapsto (cs,\bot)$ and $(p,c) \mapsto (\bot,cs)$.
\end{framed}

\noindent However, as we will show, this protocol  cannot be  simulated, in spite of the presence of a leader,  even in the simplest of the omissive models.

\subsection{Impossibility  with Finite Memory \label{impossleader}}
We investigate what happens when we introduce a distinguished leader node, but  we restrict the memory of agents to be bounded by some function of $|A|$.
We show our impossibility results directly for the {\sf T}$_3$ omissive model. The results clearly carry over to all the less powerful omissive models.


\begin{definition}{\bf ({Omission-Recurrent Configuration})} Let $\{\ell,a\}$ be a system of two agents, where $\ell$ is the leader, and let $C=(q_\ell,q_a)$ be a configuration. Suppose that there exists a finite non-empty sequence of interactions $I=(i_1,i_2,\ldots,i_t)$, where $i_1$ is omissive on both sides, such that, if $I$ is executed according to the transition rules of ${\cal P}$ starting from configuration $C$, eventually the state of $\ell$ is again $q_\ell$. Then, if $\ell$ is the starter (respectively, the reactor) of $i_1$, we say that $C$ is a \emph{starter-omission-recurrent configuration} (respectively, an \emph{reactor-omission-recurrent configuration}) for protocol ${\cal P}$.
\label{def:rec}
\end{definition}
Note that, in the above definition, since $i_1$ is omissive on both sides, the system transitions into configuration $(o(q_\ell),h(q_a))$ or $(h(q_\ell),o(q_a))$ after executing $i_1$.

\begin{lemma} \label{lemma:rec}
Let $\mathcal{S}(\mathcal{P})$ be a simulator having a finite number $k$ of states in total and working under the $\diamond${\bf NO} adversary. Let a system of two agents $\{\ell,a\}$ be given, where $\ell$ is the leader. Let $C_0$ be an initial configuration for $\ell$ and $a$, and let $I=(i_1,i_2,\ldots)$ be an infinite sequence of interactions with no omissions between $\ell$ and $a$ such that, if $\mathcal{S}(\mathcal{P})$ is executed according to $I$ starting from $C_0$, then the execution is globally fair. Then there exists a finite sequence of interactions $I'=(i'_1,i'_2,\ldots\,i'_t)$ with the following properties.
\begin{enumerate}
\item[(1)] $I'$ is obtained by introducing at most $k$ omissive interactions into an initial finite sub-sequence of $I$. All the omissive interactions of $I'$ are omissive on both sides.
\item[(2)] If $\ell$ and $a$ execute $\mathcal{S}(\mathcal{P})$ according to $I'$ starting from $C_0$, they both do a simulated state transition (according to $\delta_{\cal P}$).
\item[(3)] Suppose that $\ell$ and $a$ execute $\mathcal{S}(\mathcal{P})$ according to $I'$ starting from $C_0$, and let $C_j$ be the configuration of $\ell$ and $a$ immediately before executing interaction $i'_{j+1}$. Then, if $i'_{j+1}$ is not omissive and has $\ell$ as the starter (respectively, reactor), $C_j$ is starter-omission-recurrent (respectively, reactor-omission-recurrent) for $\mathcal{S}(\mathcal{P})$.
\end{enumerate}
\end{lemma}
\begin{proof}
First we will insert at most $k$ omissive interactions into $I$ (thus building an infinite sequence that satisfies property~(1)) in such a way as to satisfy property~(3). Then we will choose $t$ so as to satisfy property~(2).

We will construct $I'$ incrementally by an inductive procedure. Suppose we have constructed $I'$ up to $i'_j$, and let $C_{j}$ be the configuration of $\ell$ and $a$ after executing the first $j$ interactions of $I'$ starting from $C_0$ (the base case is with $j=0$). Suppose also that the partial sequence $(i'_1,\ldots,i'_j)$ has been obtained by adding some omissive interactions to some initial sub-sequence of $I$, say $(i_1,\ldots,i_{v(j)})$ (with $v(0)=0$, i.e., in the base case the sub-sequence is empty). Let $i_{v(j)+1}$ have $\ell$ as its starter (respectively, reactor). Then, if $C_j$ is starter-omission-recurrent (respectively, reactor-omission-recurrent), we set $i'_{j+1}=i_{v(j)+1}$ and $v(j+1)=v(j)+1$. Otherwise, we let $i'_{j+1}$ be omissive on both sides with $\ell$ as the starter (respectively, reactor), and we set $v(j+1)=v(j)$.

We claim that, if we continue this process indefinitely, we put at most $k$ omissive interactions in $I'$. Indeed, suppose that an omissive interaction $i'_{j}$ with $\ell$ as the starter (respectively, reactor) has been inserted in $I'$, which means that $C_{j-1}$ is not a starter-omission-recurrent (respectively, reactor-omission-recurrent) configuration. By definition of omission-recurrent configuration, $\ell$ will never get the simulated state it had in $C_{j-1}$ after executing $i'_{j}$. It follows that the same state will contribute to the addition of at most one omissive interaction to $I'$. Since the possible states for $\ell$ are at most $k$, there can be at most $k$ omissive interactions in $I'$.

We now have to decide when to stop the incremental construction of $I'$. Recall that the execution of $I$ starting from $C_0$ is globally fair, and observe that adding finitely many omissive interactions to it preserves its global fairness. So, by definition of simulator, the simulated states of $\ell$ and $a$ also change in a globally fair way as they execute $I'$. Therefore, at some point, they will conclude a simulated interaction, changing their local states according to $\delta_{\cal P}$. At this point we stop the construction of $I'$, obtaining a sequence of finite length $t$ that satisfies all three properties.
\end{proof}

\begin{theorem}\label{th:ledmemfinimp}
A system of agents, each of which has a finite amount of memory, cannot simulate every two-way protocol in the {\sf T$_3$} model (hence in all the \emph{omissive} models), even with the presence of a leader and under the  $\diamond${\bf NO} adversary. 
 \end{theorem}
\begin{proof}
We will show that the Pairing Protocol ${\cal P}_{IP}$ cannot be simulated in {\sf T$_3$}. Suppose by contradiction that there is a simulator $\mathcal{S}(\mathcal{P}_{IP})$ for it, and let us consider a system of two agents $\ell$ and $a$, where $\ell$ is the leader. Let $C_0$ be an initial configuration in which $\ell$ has simulated state $p$ and $a$ has simulated state $c$, and let $I$ be an omission-less infinite sequence of interactions for the two agents whose execution starting from $C_0$ is globally fair. According to ${\cal P}_{IP}$, eventually $I$ will make both $\ell$ and $a$ change simulated state to $\bot$ and $cs$, respectively.

Let us apply Lemma~\ref{lemma:rec} to $\mathcal{S}(\mathcal{P}_{IP})$, $C_0$, and $I$, which yields a sequence of interactions $I'=(i'_1,\ldots,i'_t)$ of length $t$, some of which are omissive on both sides. The sequence $I'$ guarantees that both agents will change simulated state to $cs$ and $\bot$, as per property~(2). Let $C_j$ be the configuration of the two agents after executing the first $j$ interactions of $I$ starting from configuration $C_0$.

Now we construct a larger system of agents: $\{\ell, a, b_1, b_2, \ldots, b_m, d\}$, where $m=2^t$. Let $C'_0$ be the configuration of this system in which $\ell$ and $a$ have the same state as in $C_0$ and all other agents have the same state as $a$. We will show how to modify $I'$ by inserting some extra interactions in it, obtaining an expanded sequence $I''$ that involves also the other members of the system (as opposed to only $\ell$ and $a$). We will then show that executing $I''$ makes the simulator behave in a way that is not compatible with the Pairing Problem.

We will construct $I''$ inductively by inserting a sequence of interactions before each interaction $i'_j\in I'$. Say we have already done so up to $i'_j$, with $0\leq j<t$ (the base case being with $j=0$). Executing $I''$ up to this point from configuration $C'_0$ makes the system reach configuration $C'_j$. Let $w(j)=m/2^j$, and suppose that agents $a$, $b_1$, $b_2$, \ldots, $b_{w(j)}$ have the same state in $C'_j$ (this is certainly true for $j=0$). Also suppose that, in $C'_{j}$, $\ell$ and $a$ have the same state as in $C_{j}$ (again, this is true for $j=0$). As we construct $I''$ and $j$ increases, we will also prove that these properties are preserved.

Now consider the next interaction $i'_{j+1}$, which could be either non-omissive or omissive on both sides. Next we are going to explain what interactions we add to $I''$ between $i'_{j}$ and $i'_{j+1}$.

Suppose that $i'_{j+1}$ is omissive on both sides, and let $\ell$ be the starter and $a$ the reactor. Then we introduce in $I''$ (right after $i'_{j}$) the sequence of interactions $(d,b_1)$, $(d,b_2)$, \ldots, $(d,b_{w(j)/2})$, omissive on both sides. Finally we introduce $i'_{j+1}$ into $I''$. If $a$ is the starter and $\ell$ the reactor, we insert the same interactions, but with starter and reactor exchanged. It is immediate to see that, after executing these interactions from configuration $C'_j$, agents $a$, $b_1$, $b_2$, \ldots, $b_{w(j+1)}$ have equal states. Indeed, they have the same state in $C'_J$, and then they all execute one omissive interaction as starters or as reactors. Moreover, in $C'_{j+1}$, $\ell$ and $a$ have the same state as in $C_{j+1}$.

Suppose that $i'_{j+1}$ is not omissive, and let $\ell$ be the starter and $a$ the reactor. By property~(3) of $I'$, configuration $C_j$ is starter-omission-recurrent. Let $q_\ell$ be the state of $\ell$ in $C_j$ (and therefore in $C'_j$). By definition of starter-omission-recurrent configuration, there exists a sequence of interactions $I^\ast=(i^\ast_1,i^\ast_2,\ldots,i^\ast_z)$, with $i^\ast_1$ omissive on both sides and having $\ell$ as the starter and $a$ as the reactor, such that the state of $\ell$ becomes $q_\ell$ again if $I^\ast$ is executed from configuration $C_j$ (and hence from $C'_j$). Note that the same happens if the partner of $\ell$ in the interactions of $I^\ast$ is not $a$ but any of the $b_x$'s, with $1\leq x\leq w(j)$, since all these agents have the same state in $C'_j$ by inductive hypothesis. Given these premises, we introduce in $I''$ (right after $i'_{j}$) the following interactions.
\begin{itemize}
\item For all $1\leq x\leq w(j)/2$, we insert:
\begin{itemize}
\item the interaction $(\ell,b_x)$, omissive on $\ell$'s side;
\item the interaction $(d,b_{x+w(j)/2})$, omissive on both sides;
\item the sequence of interactions $(i^\ast_2,i^\ast_3,\ldots,i^\ast_z)$, with $b_{x+w(j)/2}$ as $\ell$'s partner instead of $a$.
\end{itemize}
\item Finally, we insert $i'_{j+1}$.
\end{itemize}
The case in which $\ell$ is the reactor of $i'_{j+1}$ and $a$ the starter is handled in a similar fashion, but we exchange starter and reactor in the interactions that we add to $I''$. Suppose now that the system executes the above sequence of interactions starting from configuration $C'_j$, and let us focus on the sub-system consisting of $\ell$ and $b_{x+w(j)/2}$. The agent $\ell$ starts in state $q_\ell$ and then executes an omissive interaction, while $b_{x+w(j)/2}$ executes another omissive interaction. Together, these two interactions have the same effect on $\ell$ and $b_{x+w(j)/2}$ as $i^\ast_1$ (note that the partners of $\ell$ and $b_{x+w(j)/2}$ were irrelevant, because the interactions were omissive on their sides). Then $\ell$ and $b_{x+w(j)/2}$ execute all the interactions of $I^\ast$ except $i^\ast_1$. Since $b_{x+w(j)/2}$ started in the same state as $a$, by definition of $I^\ast$ it follows that the state of $\ell$ is reset again to $q_\ell$ after this sequence. By induction, this is true for all $1\leq x\leq w(j)/2$. Finally, the state of $\ell$ correctly changes according to $i'_{j+1}$ as it interacts with $a$. On the other hand, $a$ and all the $b_x$'s see $\ell$ exactly once when it is in state $q_\ell$ and, since they have the same state in $C'_j$, they also have the same state in $C'_{j+1}$.

We have shown that agents $a$, $b_1$, $b_2$, \ldots, $b_{w(j)}$ have the same state in $C'_j$ for all $1\leq j\leq t$. In particular, for $j=t$, we have that $w(j)=m/2^t=1$, which means that $a$ and $b_1$ have the same state in $C'_t$. In turn, the simulated state of $a$ in $C'_t$ is the same as in $C_t$, i.e., $cs$. Since at the beginning there was only one agent with simulated state $p$ (i.e., $\ell$), and now we have two agents with simulated state $cs$, we have violated the safety property of the Pairing Problem, meaning that $\mathcal{S}(\mathcal{P}_{IP})$ cannot be a simulator.

As there is only a finite number of omissions in $I''$, this sequence of interactions can be extended to an infinite one with the addition of non-omissive interactions, which is compatible with the $\diamond${\bf NO} adversary.
\end{proof}

\subsection{Impossibility  with Infinite Memory  \label{sec:impleaderinfmemory}}
For this case we can show that simulation is impossible in the omissive two-way model without detection, and thus in {\sf IO}. 

\begin{theorem} \label{th:infmemleader}
A system of agents, each of which has an infinite amount of memory, cannot simulate every two-way protocol in the {\sf T$_1$} model (hence in {\sf IO}), even with the presence of a leader and under the  $\diamond${\bf NO} adversary. 
 \end{theorem}
\begin{proof}
We will show that the Pairing Protocol ${\cal P}_{IP}$ cannot be simulated in {\sf T$_1$}. Let $\mathcal{S}(\mathcal{P}_{IP})$, $\ell$, $a$, $C_0$, and $I=(i_1,i_2,\ldots)$ be defined as in the first paragraph of the proof of Theorem~\ref{th:ledmemfinimp}. By definition of simulator and by the Pairing Problem, if we execute $I$ from $C_0$ according to $\mathcal{S}(\mathcal{P}_{IP})$, at some point we reach a configuration in which $\ell$ has simulated state $\bot$ and $a$ has simulated state $cs$. Say that this happens after executing $i_j$, and let $I_j=(i_1,i_2,\ldots,i_j)$.

Let us now extend the system with a third agent $b$, initially having the same state as $a$ in $C_0$, and let us show how to insert interactions into $I_j$ involving $b$ as well, in order to obtain a contradictory finite sequence of interactions $I'$. Recall that $I$ is omission-less, and consider the interaction $i_x$, with $1\leq x\leq j$. If $i_x=(\ell,a)$, we insert the interaction $(\ell,b)$ right before it, with omission on $\ell$'s side. If $i_x=(a,\ell)$, we insert the interaction $(b,\ell)$ right before it, again with omission on $\ell$'s side.

Since omissions are undetectable, it is easy to see that the extended sequence $I'$ will make $\ell$ undergo the same state transitions as $I_j$ (but at half the ``speed''). On the other hand, $a$ and $b$ will always see $\ell$ in the same state and will never see each other, so they will both have the same state throughout the execution of $I'$. It follows that $a$ and $b$ will eventually have simulated state $cs$, which violates the safety property of the Pairing Problem.

Note that the  sequence $I'$ contains finitely many omissions, and therefore it can be extended to an infinite sequence that is compatible with the $\diamond${\bf NO} adversary.
\end{proof}

\begin{observation}
Since in {\sf IO} there are no omissions, the statement of Theorem~\ref{th:infmemleader} for the {\sf IO} model trivially extends to the scenario in which the number of omissions in the sequence of interactions is known in advance by the agents.
\end{observation}


\section{Simulation in Omissive Models} \label{sec:6}

In this section we are going to make use of a result that appears in~\cite{DiFIISV16} as Theorem~4.5. This theorem assumes each agent to have a unique ID, which is a non-negative integer, as part of its local state.

\begin{theorem}
Assuming {\sf IO}, unique IDs, and $\mathcal O(\log (\max\, {\rm ID}))$ bits of memory on each agent (where $\max\, {\rm ID}$ is the maximum ID in the system), there exists a simulator for every two-way protocol, even under the {\bf UO} adversary.\qed
\label{th2:simid1}
\end{theorem}

\noindent What this theorem says is that, if the agents initially have unique IDs, they can perform a simulation of any two-way protocol, even if the simulation runs in the weakest model, {\sf IO}, and against the strongest adversary, {\bf UO}.

In this section we assume the presence of a leader and we show that, in certain models, we can implement a \emph{naming algorithm}, i.e., an algorithm that assign unique IDs to all agents. Once an ID has been assigned to an agent, it cannot change. Therefore, the naming algorithm and the simulator of Theorem~\ref{th2:simid1} can be combined into a single protocol and can even run in parallel: if an agent has no ID yet, the simulator simply ignores every interaction involving this agent. By global fairness, eventually all agents will have unique IDs, and the simulation will finally involve the entire system, producing a globally fair simulated execution.

The protocols will be presented using an algorithmic style: for each interaction of the form $(a_s,a_r)$, the starter agent $a_s$ executes function {\sf Upon Event Starter sends}$()$ and the reactor agent $a_r$ executes {\sf Upon Event Reactor delivers}~{$(var^{s})$}, where $var^s$ is the variable $var$ in the local state of agent $a_s$.

\subsection{Naming Algorithm with Infinite Memory}
If the leader has infinite memory, it can implement a simple naming algorithm under certain models.  Since Theorem~\ref{th:infmemleader} already states the impossibility of simulation under models {\sf T$_1$} and {\sf IO}, we will assume model {\sf I$_1$} or model {\sf I$_2$}. Constructing a simulator for these models will imply the existence of a simulator for all other models except {\sf T$_1$} and {\sf IO} (refer to Figure~\ref{id:algorwq2}).

\begin{theorem}
\label{th2:leaderinfmemsim}
Assuming {\sf I$_1$} or {\sf I$_2$}, the presence of a leader, and an infinite amount of memory on each agent, there exists a simulator for every two-way protocol, even under the {\bf UO} adversary.
\end{theorem}
\begin{proof}
By the above discussion, it suffices to implement a naming algorithm. Each agent has a local variable $my\_ID$ and the leader also has a second variable $next\_ID$. The leader has $my\_ID=0$ and $next\_ID=1$, while all other agents initially have $my\_ID=\bot$. Every time the leader detects the proximity of another agent (i.e., it applies function $g$ to its own state), it increments the variable $next\_ID$. Every time an agent sees the state of the leader (i.e., it applies function $f$), it sets its own $my\_ID$ to the value found in the leader's $next\_ID$ variable.

Since the leader increments $next\_ID$ every time it is involved in an interaction (even if the interaction is omissive on the other side), no two agents can get the same ID. Moreover, by global fairness, all agents will eventually see the leader in a non-omissive interaction, and will therefore acquire an ID.
\end{proof}

\subsection{Naming Algorithms with Knowledge on Omissions}

Now we assume that agents have only a finite amount of memory, but they know in advance a finite upper bound $L$ on the number of omission faults that the adversary is going to insert \emph{in interactions that involve the leader}. Note that the adversary can still be {\bf UO} even if only finitely many omissive interactions involve the leader.

\subsubsection{Naming Algorithm for {\sf I}$_1$ and {\sf I}$_2$}

We refine the naming algorithm of Theorem~\ref{th2:leaderinfmemsim} to cope with the fact that now memory is bounded by a function of $L$ and the size of the system, $n$. It is worth mentioning that the precise value of $n$ is not known to the agents, and $L$ is only an upper bound on the number of omissions involving the leader, not necessarily the exact number.

\begin{theorem}\label{upbound:1}
Assuming {\sf I$_1$} or {\sf I$_2$}, the presence of a leader, knowledge of an upper bound $L$ on the number of omission failures in interactions that involve the leader, and $\Theta(L \log nL)$ bits of memory on each agent (where $n$ is the number of agents), there exists a simulator for every two-way protocol, even under the {\bf UO} adversary.
\end{theorem}
\begin{proof}
We implement the naming algorithm presented in Figure~\ref{p:naming1}. Compared to the algorithm of Theorem~\ref{th2:leaderinfmemsim}, here the leader has an array of $L+1$ $next\_ID$ variables, as opposed to only one. This array is initialized to $[1,2,\ldots,L+1]$ and, when an ID is assigned, the corresponding entry of the array will be incremented by $L+1$, so that no two equal IDs are ever be generated.

All entries of $next\_ID$ are initially \emph{unlocked}: this information is stored in the leader's Boolean array $locked$. The \emph{active ID} is defined as the unlocked entry of $next\_ID$ having minimum index, if there is any (line~\ref{a1l9}). This is the ID that will tentatively be assigned next. Whenever the leader detects the proximity of another agent (i.e., it executes function $g$ on its own state, or function {\sf Upon Event Starter sends} in the algorithm of Figure~\ref{p:naming1}), it locks the active entry of $next\_ID$ (line~\ref{a1l11}). The purpose of locking an entry of $next\_ID$ (as opposed to just incrementing it as in Theorem~\ref{th2:leaderinfmemsim}) is that the leader cannot allow its value to grow indefinitely, because now memory is limited. Instead, the leader will make the entry temporarily inactive, and will keep it on hold until it gathers more information in the following interactions.

On the other hand, if an agent $a$ sees the leader (i.e., it executes function $f$ or function {\sf Upon Event Reactor delivers} in Figure~\ref{p:naming1}), and $a$ does not have an ID yet, then it assigns itself the active ID from the leader's $next\_ID$ variable (line~\ref{a1l29}). So, the next time the leader sees $a$, it will read its new ID and it will know that the corresponding entry of $next\_ID$ can be unlocked (line~\ref{a1l22}) and its value can be incremented by $L+1$ (line~\ref{a1l23}).

It may happen that the leader is involved in an omissive interaction, and therefore the entry of $next\_ID$ that it locks will never be unlocked again. However, this can happen at most $L$ times, while the array has $L+1$ entries.

This is not sufficient yet, because the same agent $a$ may see the leader multiple times in a row and cause all entries of $next\_ID$ to become locked. If $a$ only stores one ID, it will have no way to tell the leader that more than one entry of $next\_ID$ has to be unlocked. This is why $a$ also has a variable called $redundant$, which is a Boolean array that will store information on all the active entries of $next\_ID$ that $a$ sees after receiving an ID. So, if the agent $a$ already has an ID and it sees the leader again, it sets to $true$ the entry of $reduntant$ corresponding to the active ID of the leader (line~\ref{a1l31}).

Now, suppose that the leader sees that $a$ has an entry of $redundant$ set to $true$. This implies that the corresponding entry of $next\_ID$ is currently locked and should be unlocked. However, this cannot be done right away: the leader wants to give $a$ an ``acknowledgment'', so that $a$ will set the entry of $redundant$ to $false$ first. This is to prevent the scenario in which the entry of $next\_ID$ gets unlocked, becomes active, another agent $b$ sees it, and takes it as its own ID. If then the leader sees $a$ again (still with $redundant$ on $true$), it will unlock the entry of $next\_ID$. Then perhaps yet another agent $c$ will see the leader, getting the same ID as $b$.

To prevent such an incorrect behavior, the leader has another variable array called $waiting$, in which it stores the IDs of the agents that should reset their $redundant$ variables. So, when the leader sees that $a$ has some entry of $redundant$ set to $true$, it stores the (unique) ID of $a$ in the corresponding entry of $waiting$ (line~\ref{a1l17}). Then, when $a$ sees the leader again and reads its own ID in the $waiting$ array, it knows that it has to set to $false$ the corresponding entries of $redundant$ (line~\ref{a1l33}). Finally, when the leader sees $a$ again and notices that the entry of $redundant$ has been set to $false$, it can reset the corresponding entry of $waiting$ (line~\ref{a1l19}) and unlock the entry of $new\_ID$ (line~\ref{a1l20}).

The fact that the algorithm does not give the same ID to two different agents follows from the observation that at most one agent can keep an entry of $new\_ID$ locked at any given time, which in turn follows from the way the two variables $redundant$ and $waiting$ function together. If no omission occurs and the leader is observed by some agent $a$, then $a$ will store information about the currently active ID. If $a$ takes this ID for itself, that entry of $next\_ID$ will be incremented before any other agent can get the same ID. If $a$ has already an ID, the entry of $next\_ID$ will remain locked until $a$ has reset its own $redundant$ variable. Moreover, the fact that the algorithm will eventually assign every agent an ID immediately follows from the global fairness of the adversarial scheduler.

Since the IDs in the $next\_ID$ array increase by $L+1$ every time one is assigned, and since there are $n$ agents in total, the value of every ID is $\mathcal O(nL)$. Hence, $\mathcal O(L\log nL)$ bits of memory are required to store each agent's arrays, and $\mathcal O(\log nL)$ more bits are required to run the simulator of Theorem~\ref{th2:simid1}. The total amount of memory needed per agent is therefore $\mathcal O(L\log nL)$ bits.
\end{proof}

\begin{figure}
\scriptsize
\begin{framed}
\begin{algorithmic}[1]

\State {\sf Variables}
\State $my\_ID$\Comment{the leader has this variable initialized to $0$, non-leaders to $\bot$} 
\State $next\_ID[]:=[1,2,\ldots,L+1]$ 		      \Comment{leader variable} 
\State $locked[]:=[false,false,\ldots,false]$   \Comment{leader variable} 
\State $waiting[]:=[\bot,\bot,\ldots,\bot]$   \Comment{leader variable} 
\State $redundant[]:=[false,false,\ldots,false]$ \Comment{non-leader variable} 

\\
\State {\sf Upon Event Starter sends}() \label{a1l7}
\If{$my\_ID=0$} \Comment{I am the leader} \label{a1l8}
\State $j:=\min\{j\mid locked[j]=false, L+2\}$\label{a1l9}
\If{$j < L+2$} \label{a1l10}
\State  $locked[j]:=true$ \label{a1l11}
\EndIf
\EndIf

\\
\State {\sf Upon Event Reactor delivers} {($my\_ID^{s}, next\_ID^{s}[], locked^{s}[], waiting^{s}[], redundant^{s}[]$)}  \label{a1l13}
\If{$my\_ID=0$} \Comment{I am the leader} \label{a1l14}
\ForAll{$j \in \{1,2,\ldots,L+1\}$}\label{a1l15}
\If{$redundant^{s}[j]=true$}\label{a1l16}
\State  $waiting[j]:=my\_ID^{s}$\label{a1l17}
\ElsIf{$waiting[j]=my\_ID^{s}$}\label{a1l18}
\State $waiting[j] := \bot$ \label{a1l19}
\State $locked[j]:=false$\label{a1l20}
\EndIf
\EndFor
\If{$\exists j, next\_ID[j]=my\_ID^{s}$}\label{a1l21}
\State $locked[j]:=false$\label{a1l22}
\State $next\_ID[j]:=next\_ID[j]+L+1$ \label{a1l23}
\EndIf

\Else  \Comment{I am not the leader}\label{a1l24}
\If{$my\_ID^{s}=0$} \Comment{my partner is the leader}\label{a1l25}
\State $j=\min\{j \mid locked^{s}[j]=false, L+2\}$\label{a1l26}
\If{$j < L+2$} \label{a1l27}
\If{$my\_ID=\bot$}\label{a1l28}
\State $my\_ID:=next\_ID^{s}[j]$ \label{a1l29}
\Else\label{a1l30}
\State $redundant[j]:=true$ \label{a1l31}
\EndIf
\EndIf
\If{$my\_ID \neq \bot\, \land\, \exists j, waiting^{s}[j] = my\_ID$}\label{a1l32}
\State $redundant[j]:=false$\label{a1l33}
\EndIf
\EndIf
\EndIf

\end{algorithmic}
\caption{Naming algorithm for {\sf I}$_1$ and {\sf I}$_2$ with knowledge on omissions, used in Theorem~\ref{upbound:1} \label{p:naming1}}
\end{framed}
\vspace{0.5cm}
\begin{framed}
\begin{algorithmic}[1]

\State {\sf Variables}
\State $role$ \Comment{the leader has this variable initialized to $leader$, non-leaders to $available$} 
\State $state_{\mathcal{P}}:= initial\_state_{\mathcal{P}}$
\State $token:=\bot$

\\
\State {\sf Upon Event Starter sends}()\label{a2l6}
\State $token:=\bot$\label{a2l7}
\If{$role=leader$}\label{a2l8}
	\State $role:=available$\label{a2l9}
\ElsIf{$role=moving$}\label{a2l10}
	\State $role:=available$\label{a2l11}
\ElsIf{$role=starter$}\label{a2l12}
	\State $role:=pending$\label{a2l13}
\EndIf

\\
\State {\sf Upon Event Reactor delivers} {$(role^s,state^{s}_{\mathcal{P}},token^{s})$} \label{a2l15}
\State $token:=token^{s}$\label{a2l16}
\If{$role^s= leader$}\label{a2l17}
    \State $role:=moving$\label{a2l18}
\ElsIf{$role^s= moving$}\label{a2l19}
    \State $role:=starter$\label{a2l20}
\ElsIf{$role^s=starter$}\label{a2l21}
	\State $token:=state_{\mathcal{P}}$\label{a2l22}
    \State $state_{\mathcal{P}}:=f_r(state^{s}_{\mathcal{P}}, state_{\mathcal{P}})$\label{a2l23}
\ElsIf{$role= pending \, \wedge \, token \neq \bot$}\label{a2l24}
	\State $state_{\mathcal{P}}:=f_s(state_{\mathcal{P}},token)$\label{a2l25}
	\State $role:=leader$\label{a2l26}
	\State $token:=\bot$\label{a2l27}
\EndIf

\end{algorithmic}
\caption{Simulation protocol for {\sf IT} with finite memory, used in Theorem~\ref{leader:token} \label{p:weakalgorithm}}
\end{framed}

\restoregeometry
\end{figure}

\subsubsection{Naming Algorithm for {\sf T}$_1$}

Observe that the previous naming algorithm does not work for model {\sf T}$_1$, and Theorem~\ref{th:infmemleader} does not hold when some kind of upper bound on omissions is known.

\begin{theorem}\label{upbound:2}
Assuming {\sf T$_1$}, the presence of a leader, knowledge of an upper bound $L$ on the number of omission failures in interactions that involve the leader, and $\Theta(L \log nL)$ bits of memory on each agent (where $n$ is the number of agents), there exists a simulator for every two-way protocol, even under the {\bf UO} adversary.
\end{theorem}
\begin{proof}
It is sufficient to give a naming algorithm. We modify the one used in Theorem~\ref{th2:leaderinfmemsim} to work in {\sf T$_1$} with ${\cal O}(L \log nL)$ memory. The leader has the same two local variables, but each other agent has an array $my\_ID$ of $L+1$ local variables, each of which is initially set to $\bot$. If an agent $a$ sees the leader, and the local array $my\_ID$ of $a$ has some entries still set to $\bot$, then $a$ changes one of them from $\bot$ to the value of the leader's $next\_ID$ variable. On the other hand, if the leader sees an agent whose local array $my\_ID$ has some entries still set to $\bot$, it increments $next\_ID$. When all entries of an agent's array $my\_ID$ have been set, the entire array is taken as the agent's ID.

Since the execution is globally fair, eventually all agents will have their $my\_ID$ array completely set. Observe that, whenever the leader increments $next\_ID$, there is an agent $a$ that removes one occurrence of $\bot$ from its local array $my\_ID$, unless the interaction is omissive on $a$'s side. But there can only be $L$ such omissive interactions, which means that the maximum value of $next\_ID$ will be $\mathcal O(nL)$. So, $\Theta(L\log nL)$ bits of memory are enough for an agent to store its local array $my\_ID$. By Theorem~\ref{th2:simid1}, combining this naming algorithm with the simulator does not require more than $\Theta(L\log nL)$ bits of memory on each agent.

Suppose for a contradiction that two agents receive equal IDs. Therefore, both agents have observed the leader $L+1$ times, and at the $j$th observation both agents must have read the same value in variable $next\_ID$, for all $1\leq j\leq L+1$. So, the leader has failed to increment $next\_ID$ for at least $L+1$ times, implying that there have been $L+1$ omissive interactions involving the leader, which contradicts the theorem's assumptions. Thus all agents receive distinct IDs, and the naming algorithm is correct.
\end{proof}

\section{Simulation for {\sf IT} \label{sec:5}}

Notice that {\sf IT} is the only finite-memory model for which the impossibility result of Theorem~\ref{th:ledmemfinimp} does not hold (see Figure~\ref{id:algorwq2}). It turns out that in this model we can implement a simulator that sequentializes the simulated two-way interactions via a token-passing technique.

\begin{theorem}
Assuming {\sf IT}, the presence of a leader, and a constant amount of memory on each agent, there exists a simulator for every two-way protocol, even under the {\bf UO} adversary. \label{leader:token}
\end{theorem}
\begin{proof}
The simulation algorithm is reported in Figure~\ref{p:weakalgorithm}. Suppose we are given a two-way protocol $\mathcal P$ whose transition function is $\delta_\mathcal P(a_s,a_r)=(f_s(a_s,a_r),f_r(a_s,a_r))$. In our simulator, each agent has a local variable called $state_\mathcal P$, which is the state of $\mathcal P$ that the agent is simulating, plus an auxiliary variable $role$, which is used to coordinate the simulation. Initially, the role of one agent is $leader$, while all others are $available$. When the leader meets another agent $a$, the leadership is ``transferred'' to $a$: the role of the leader becomes $available$ (line~\ref{a2l9}), and the role of $a$ becomes $moving$ (line~\ref{a2l18}). Note that the leader does not have to see the state of $a$ to perform this operation.

The next agent $b$ that sees $a$ becomes the starter of a new simulated interaction: the role of $b$ becomes $starter$ (line~\ref{a2l20}) and the role of $a$ becomes $available$ again (line~\ref{a2l11}). Now, the first agent $c$ that observes $b$ becomes the reactor of the simulated interaction: it executes function $f_r$ on its own simulated state using $b$'s simulated state as part of the input (line~\ref{a2l23}), while $b$'s role becomes $pending$ (line~\ref{a2l13}).

Now, in order for $b$ to perform its side of the simulated transition, it has to retrieve the simulated state that $c$ had before transitioning. To deliver this information to $b$, the agent $c$ stores its own simulated state in a variable called $token$ before performing the transition (line~\ref{a2l22}). Now, as soon as an agent sees $c$, it copies the token (line~\ref{a2l16}), while $c$ erases its own copy (line~\ref{a2l7}). This token circulation protocol is executed until the token reaches $b$.

When $b$ finally obtains the token, it uses it as part of the input to function $f_s$ and changes its simulated state accordingly (line~\ref{a2l25}). Now both sides of the simulated transition have been performed correctly, and $b$ resets its role to $leader$ (line~\ref{a2l26}) and destroys the token (line~\ref{a2l27}). At this point we have exactly one agent whose role is $leader$, while all other agents have role $available$, as we had at the beginning. The next steps of the simulation are thus performed in the same fashion.

Note that, if the two-way protocol $\mathcal P$ has a constant number $k$ of states, then our simulator has $\mathcal O(k^2)$ states, independently of the size of the system.

The correctness of the simulator can be proven by observing that, due to the uniqueness of the leader, there is at most one pending transition at all times. Moreover, any agent in the system can become the starter (including the leader itself), thanks to the extra step that creates an agent with role $moving$: any non-leader agent can become $moving$, and then any agent (including the original leader) can become $starter$. Note that this is not true if the system consists of $n=2$ agents (in this case the leader will necessarily become the starter of every simulated interaction), but recall that the definition of simulator assumes that $n>2$ (cf. Section~\ref{s:simulation}).

We have to prove that the system will perform infinitely many simulated interactions, in such a way that the simulated execution is globally fair. By the global fairness of the simulator, a token will certainly be created and will be passed around by the agents; again by global fairness, the token will eventually reach the agent with $role=starter$, and the simulated interaction will be concluded. The global fairness of the entire simulated execution also follows immediately from the global fairness of the simulator itself.
\end{proof}

\bibliographystyle{plain}

\end{document}